\documentclass[12pt, onecolumn]{article}

\usepackage[dvipdfmx]{graphicx}
\usepackage[utf8]{inputenc}
\usepackage[T1]{fontenc}
\usepackage[numbers]{natbib}
\usepackage{amsmath}
\usepackage{amsthm}
\usepackage{amssymb}
\usepackage{algorithm}
\usepackage{algorithmic}
\usepackage{ecltree}
\usepackage{enumerate}
\usepackage{xcolor}
\usepackage{comment}
\usepackage{tikz}
\usepackage{wrapfig}
\usepackage{mathrsfs}
\usepackage{mathtools}
\usepackage[margin=25truemm]{geometry}
\mathtoolsset{showonlyrefs=true}

\usepackage{bm}
\usepackage{bbm}
\usepackage{float}
\usepackage{cases}
\usepackage{longtable}
\usepackage{diagbox}
\usepackage{booktabs}

\newtheorem{proposition}{Proposition}
\newtheorem{lemma}[proposition]{Lemma}
\newtheorem{theorem}[proposition]{Theorem}

\theoremstyle{definition}
\newtheorem{definition}[proposition]{Definition}

\newtheorem{ruleset}[proposition]{Ruleset}

\newcommand{\NN}{{\mathcal{N}}}
\newcommand{\PP}{{\mathcal{P}}}

\newcommand{\uint}{\mathbb{Z}_{\geq 0}}

\DeclareMathOperator*{\mex}{mex}

\newcommand{\legalS}{\mathcal{S}_{\mathrm{legal}}}
\newcommand{\legalH}{\mathcal{H}_{\mathrm{legal}}}
\newcommand{\goodH}{\mathcal{H}_{\mathrm{good}}}
\newcommand{\subH}{\mathcal{H}}
\newcommand{\gpg}{g_{\mathrm{PG}}}
\newcommand{\gat}{g_{\mathrm{AT}}}

\newcommand{\eqlab}[2]{\overset{(\mathrm{#1})}{#2}}

\makeatletter

\@addtoreset{equation}{section}
\makeatother

\title{The Two-Level Poset Game is \textsf{PSPACE}-Complete}
\author{Kengo Hashimoto}
\date{}

\begin{document}
\maketitle

\begin{abstract}
Poset games form a fundamental class of combinatorial games that includes classical games such as Nim and Chomp.
Determining the outcome of a poset game is known to be \textsf{PSPACE}-complete even for posets of height at most three.
However, the exact computational complexity of \textsc{Two-Level Posetgame}, played on posets of height at most two, has remained an open problem for over a decade.
We resolve this problem by proving that \textsc{Two-Level Posetgame} is \textsf{PSPACE}-complete via a polynomial-time reduction from \textsc{Avoid True} on positive CNF formulas.
\end{abstract}

\section{Introduction}
Combinatorial games (i.e., two-player games without chance elements or hidden information) played on structured combinatorial objects have been a major topic in combinatorial game theory.
\emph{Poset games} form a large class of combinatorial games including many classical and fundamental rulesets such as \textsc{Nim} \cite{Bou1901}, \textsc{Chomp} \cite{Gale1974}, \textsc{Subset Take-Away} \cite{Gale1982}, and \textsc{Green Hackenbush} on a rooted tree \cite{BCG18}.
A poset game is played as follows.

\begin{ruleset}[poset game]
Given a finite partially ordered set (poset) $(P, \le)$, two players alternately choose an element $x \in P$.
The chosen element $x$ and all elements $y \in P$ such that $y > x$ are removed from $P$.
Under the normal play convention, the player who cannot make a move loses.
\end{ruleset}

Grier~\cite{Grier2013} proved that the problem of determining the \emph{outcome} (i.e., which player has a winning strategy) of an arbitrary poset game is \textsf{PSPACE}-complete via a reduction from \textsc{Node Kayles}.
On the other hand, while the construction in Grier's reduction produces posets of height at most three, the exact computational complexity of \textsc{Two-Level Posetgame}, the poset game on a poset of height at most two, has remained an open problem for over a decade~\cite{FGG22, Fenner2013}.

In this paper, we resolve this open problem by proving that \textsc{Two-Level Posetgame} is \textsf{PSPACE}-complete as follows.
\begin{theorem}
\label{thm:main}
The problem of determining the outcome of \textsc{Two-Level Posetgame} is \textsf{PSPACE}-complete.
\end{theorem}

To establish Theorem~\ref{thm:main}, we present a polynomial-time reduction from
the problem of determining the outcome of \textsc{Avoid True} on positive CNF formulas, which is known to be \textsf{PSPACE}-complete~\cite{Schaefer1978}.

The remainder of this paper is organized as follows.
In the remainder of this section, we review other related work on \textsc{Two-Level Posetgame}.
In Section~\ref{sec:preliminaries}, we introduce basic definitions and foundational properties: Sprague--Grundy theory, \textsc{Two-Level Posetgame}, and \textsc{Avoid True} on positive CNF formulas.
In Section~\ref{sec:proof}, we provide the proof of Theorem~\ref{thm:main}.

\subsection{Related Work}
Fenner et al.~\cite{Fenner2013} provided formulas for Grundy values for restricted classes of \textsc{Two-Level Posetgame}.
In particular, their results give a formula for the case in which every hyperedge has odd cardinality, which we use in our reduction.

Because a poset of \textsc{Two-Level Posetgame} consists of a lower level consisting of all minimal elements and an upper level consisting of the remaining elements, the rules and positions can be naturally rephrased in hypergraph terms: selecting a vertex removes it and all hyperedges containing it, whereas selecting a hyperedge removes only that hyperedge.
Barnard~\cite[Chapter 2]{barnard2017} and Molena~\cite{Molena2022} investigated \textsc{Two-Level Posetgame} under this hypergraph formulation.
Regarding the particular case restricted to ordinary graphs,
O'Sullivan~\cite{OSullivan2018} derived explicit formulas for Grundy values for several classes of graphs.

\section{Preliminaries}\label{sec:preliminaries}

For a positive integer $n$, let $[n] \coloneqq \{1, 2, \ldots, n\}$.

\subsection{Sprague--Grundy Theory}

All games considered in this paper are finite, impartial (i.e., the sets of moves of the two players are the same), and played under the normal play convention (i.e., the player who cannot make a move loses).
By the Sprague--Grundy Theorem~\cite{Grundy1939, Sprague1935}, every position $G$ in an impartial game is equivalent to a single Nim-pile of size $g(G)$ for some $g(G) \in \uint$.
The value $g(G) \in \uint$ is called the \emph{Grundy value} and is given recursively as
\begin{equation}
g(G) = \mex \{ g(G') \mid G \to G' \},
\end{equation}
where $G \to G'$ denotes that there exists a move from $G$ to a position $G'$, and $\mex(S) \coloneqq \min(\uint \setminus S)$ for a set $S \subsetneq \uint$.
A position $G$ is a $\PP$-position (i.e., second-player win) if and only if $g(G) = 0$; otherwise, it is an $\NN$-position (i.e., first-player win).

For two positions $G_1$ and $G_2$, the \emph{disjunctive sum} $G_1 + G_2$ of $G_1$ and $G_2$ is the position of the compound game of $G_1$ and $G_2$ where a player makes a move on exactly one of $G_1$ and $G_2$.
The Grundy value of $G_1 + G_2$ is obtained from the Grundy values of the individual components as
\begin{equation}
g(G_1 + G_2) = g(G_1) \oplus g(G_2),
\end{equation}
where $\oplus$ denotes the \emph{nim-sum} (i.e., the bitwise XOR operation).

\subsection{\textsc{Two-Level Posetgame}}

A poset $P = (V \sqcup E, \le)$ for \textsc{Two-Level Posetgame} consists of the set $V$ of minimal elements and the set $E$ of the other elements. By identifying $V$ with the set of vertices and $E$ with the set of hyperedges where $v \in e$ if and only if $v \le e$,
the rules of \textsc{Two-Level Posetgame} on $P$ translate directly as follows.

\begin{ruleset}
Initially, a hypergraph $H = (V, E)$ is given. In each turn, a player removes either a vertex or a hyperedge of $H$.
When a vertex $v$ is removed, all hyperedges containing $v$ are automatically removed as well. The player who cannot make a move because $H$ becomes empty loses.
\end{ruleset}
The hypergraphs for \textsc{Two-Level Posetgame} may contain duplicate hyperedges. \allowbreak
Throughout this paper, we adopt this hypergraph representation of \textsc{Two-Level Posetgame} described above.

The following explicit formula is known for \textsc{Two-Level Posetgame} where every hyperedge contains an odd number of vertices.
\begin{proposition}[\cite{Fenner2013}]
\label{prop:odd_formula}
Let $H = (V, E)$ be a hypergraph such that every hyperedge $e \in E$ has odd cardinality.
Then the Grundy value $g(H)$ of the initial position of \textsc{Two-Level Posetgame} on $H$ is given by
\begin{equation}
\label{eq:formula2}
g(H) = (|V| \bmod 2) \oplus 3(|E| \bmod 2).
\end{equation}
\end{proposition}

\subsection{\textsc{Avoid True} on positive CNF formulas}

The game \textsc{Avoid True} on positive CNF formulas is an impartial game played on a CNF formula
\begin{equation}
F = \bigwedge_{i \in [m]} \left( \bigvee_{x \in C_i} x \right),
\end{equation}
where $m \geq 1$ and each clause $C_i$ is a non-empty set of variables from $X = \{x_1, \dots, x_n\}$ and
contains no negative literals.
The game is played as follows.

\begin{ruleset}[\textsc{Avoid True} on positive CNF formulas]
A CNF formula over a variable set $X$ without negative literals is given.
Initially, all variables in $X$ are assigned \textsf{False}.
Two players take turns selecting a previously unchosen variable $x \in X$ and setting it to \textsf{True}, 
subject to the condition that setting $x$ to \textsf{True} does not satisfy the overall formula $F$.
In other words, at least one clause must remain unsatisfied after the move.
The first player who cannot make a move loses.
\end{ruleset}

Regarding the computational complexity of \textsc{Avoid True} on positive CNF formulas, the following result is known.
\begin{proposition}[\cite{Schaefer1978}]
\label{thm:avoid_true_pspace}
The problem of determining the outcome of \textsc{Avoid True} on positive CNF formulas is \textsf{PSPACE}-complete.
\end{proposition}

Henceforth, we refer to \textsc{Avoid True} on positive CNF formulas simply as \textsc{Avoid True}.

\section{Proof of Theorem \ref{thm:main}}
\label{sec:proof}

Since the plays of \textsc{Two-Level Posetgame} on $H = (V, E)$ terminate within $|V|+|E|$ moves and each position is represented in polynomial space in $|V|+|E|$, the problem clearly belongs to \textsf{PSPACE}.
Thus, we establish its \textsf{PSPACE}-hardness via a polynomial-time reduction from the problem of determining the outcome of \textsc{Avoid True} in Section~\ref{sec:construction}.
The proof of Lemma~\ref{lem:deviating}, which is used in this reduction, is deferred to Section~\ref{sec:proof_deviating}.

Although hypergraphs for \textsc{Two-Level Posetgame} may in general contain multiple identical hyperedges, 
our reduction does not require multiedges.
Thus, we assume that $E$ contains no duplicate hyperedges and treat $E$ as a subset of the power set of $V$.

\subsection{Proof of \textsf{PSPACE}-hardness}
\label{sec:construction}

Let $F = \bigwedge_{i \in [m]} \left( \bigvee_{x \in C_i} x \right)$ be an instance of \textsc{Avoid True} with variable set $X = \{x_1, \dots, x_n\}$.

We construct a hypergraph $H_{\emptyset} = (V(H_{\emptyset}), E(H_{\emptyset}))$ corresponding to $F$ as follows, where the subscript $\emptyset$ is chosen to align with the notation $H_S$ introduced later (i.e., corresponding to $S = \emptyset$).
\begin{itemize}
\item The vertex set is defined as
\begin{align}
V(H_{\emptyset}) \coloneqq X \sqcup \{u, a, b\} \sqcup \bigsqcup_{i\in[m]} \bigsqcup_{j \in \{0, 1\}} D_{i, j},
\end{align}
where for each clause $C_i$ and $j \in \{0, 1\}$, the set $D_{i, j}$ consists of auxiliary vertices with cardinality
\begin{equation}
|D_{i, j}| = \begin{cases}
1 & \text{if } |C_i| \text{ is odd}, \\
2 & \text{if } |C_i| \text{ is even}.
\end{cases}
\end{equation}
We define a subset of auxiliary vertices $M \subseteq V(H_{\emptyset})$ by
\begin{equation}
M \coloneqq \{a, b\} \cup \bigsqcup_{i\in[m]} \bigsqcup_{j \in \{0, 1\}} D_{i, j}.
\end{equation}

\item The hyperedge set is defined as
\begin{align}
E(H_{\emptyset}) \coloneqq \{e_M\} \sqcup \{ e_{i, j} \mid i \in [m], j \in \{0, 1\} \},
\end{align}
where $e_M \coloneqq M$, and $e_{i, j} \coloneqq C_i \cup \{u\} \cup D_{i, j}$ for each $i \in [m]$ and $j \in \{0, 1\}$.
\end{itemize}

\begin{figure}[htpb]
\centering
\includegraphics[keepaspectratio,scale=0.6]{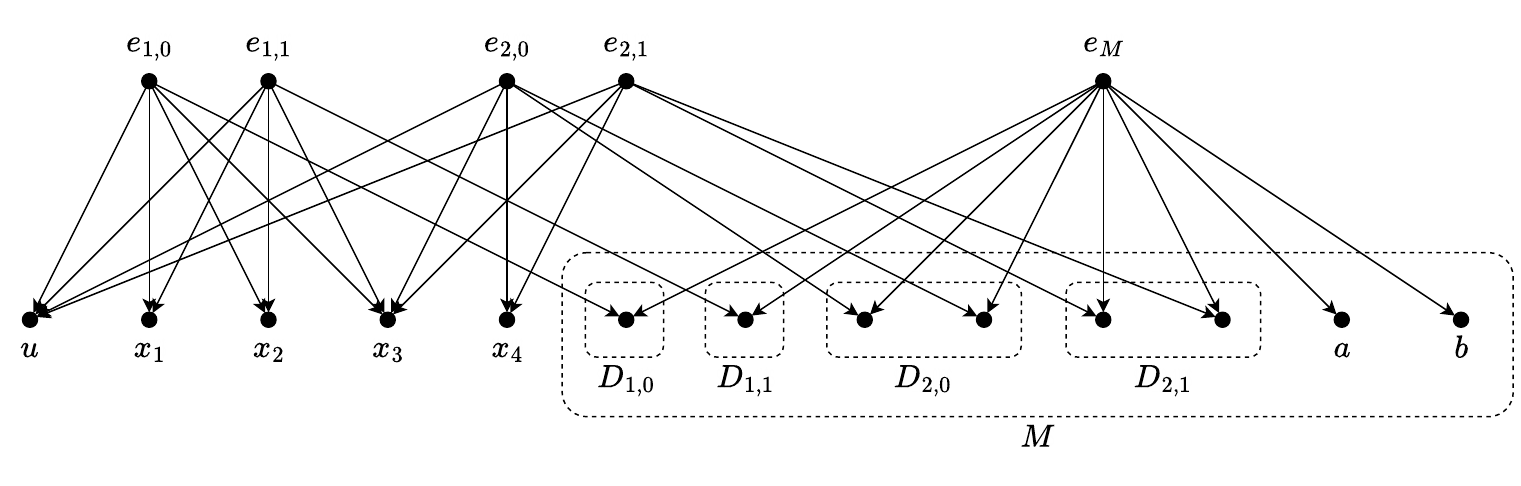}
\caption{The hypergraph $H_{\emptyset}$ corresponding to $m = 2$, $C_1 = \{x_1, x_2, x_3\}$ and $C_2 = \{x_3, x_4\}$}
\label{fig:two-level}
\end{figure}

Figure~\ref{fig:two-level} shows the hypergraph $H_{\emptyset}$ corresponding to the case $m = 2, C_1 = \{x_1, x_2, x_3\}, C_2 = \{x_3, x_4\}$; namely, $F = (x_1 \lor x_2 \lor x_3) \land (x_3 \lor x_4)$.

We refer to the hyperedges $e_{i, j}$ corresponding to clauses as the \emph{clause hyperedges}.
By construction, the edge $e_M$ has even cardinality, whereas every clause hyperedge $e_{i, j}$ has odd cardinality.

We identify a position of \textsc{Avoid True} with the set $S \subseteq X$ of variables currently assigned \textsf{True}.
A position $S \subseteq X$ of \textsc{Avoid True} reachable from the initial position $\emptyset$ via a sequence of moves is said to be \emph{legal}.
Let $\legalS$ denote the set of all legal positions of \textsc{Avoid True}; namely,
\begin{align}
\legalS \coloneqq \{S \subseteq X \mid \exists i \in [m] \text{ s.t. } C_i \cap S = \emptyset\}.
\end{align}
For $S \in \legalS$, let $\gat(S)$ denote the Grundy value of $S$ as a position of \textsc{Avoid True}; namely,
\begin{align}
\gat(S) = \mex \{ \gat(S \cup \{x\}) \mid S \cup \{x\} \in \legalS, x \in X \setminus S\}.
\end{align}

We identify a position of the \textsc{Two-Level Posetgame} with the remaining hypergraph $H$ in that position.
For a position $H$ of \textsc{Two-Level Posetgame},
let $\gpg(H)$ denote the Grundy value of $H$ as a position of \textsc{Two-Level Posetgame}; namely,
\begin{align}
\gpg(H) = \mex \{ \gpg(H') \mid H \to H'\}.
\end{align}
Let $\subH$ denote the set of all subpositions of $H_{\emptyset}$ (i.e., positions reachable from $H_{\emptyset}$ by a sequence of moves on \textsc{Two-Level Posetgame}).
For $S \in \legalS$, let $H_S$ denote the subposition of $H_{\emptyset}$ obtained by removing the vertices corresponding to $S$ from $H_{\emptyset}$;
namely, $H_S = (V(H_S), E(H_S)) \in \subH$ is defined as
\begin{align}
V(H_S) &\coloneqq V(H_{\emptyset}) \setminus S,\\
E(H_S) &\coloneqq E(H_{\emptyset}) \setminus \{e_{i, j} \mid C_i \cap S \neq \emptyset, j \in \{0, 1\}\}.
\end{align}
In particular, when $S = \emptyset$, the subposition $H_S$ coincides with $H_{\emptyset}$ itself, which is consistent with our subscript notation.
A subposition $H \in \subH$ is said to be \emph{legal} if $H = H_S$ for some $S \in \legalS$; otherwise, $H$ is said to be \emph{illegal}.
Let $\legalH$ denote the set of all legal positions of $\subH$.

The following lemma shows that from any legal position $H \in \legalH$, the set of Grundy values of the illegal positions reachable via a single move is precisely $\{0, 1, 2, 3\}$.

\begin{lemma}
\label{lem:deviating}
For every $H \in \legalH$, we have $\{\gpg(H') \mid H \to H', H' \not\in \legalH\} = \{0, 1, 2, 3\}$.
\end{lemma}
The proof of Lemma~\ref{lem:deviating} is deferred to Section~\ref{sec:proof_deviating}.

Using Lemma~\ref{lem:deviating}, we establish the connection between the Grundy values $\gat(S)$ and $\gpg(H_S)$ as follows.

\begin{lemma}
\label{lem:h_to_g}
For every $S \in \legalS$, we have $\gpg(H_S) = 4 + \gat(S)$,
where $+$ denotes the standard addition.
\end{lemma}

\begin{proof}
We prove this statement by induction on $|S|$ in decreasing order. We have
\begin{align}
\gpg(H_S)
&= \mex\{\gpg(H') \mid H_S \to H'\}\\
&= \mex\left(\{\gpg(H') \mid H_S \to H', H' \not\in \legalH\} \cup \{\gpg(H') \mid H_S \to H', H' \in \legalH\} \right)\\
&\eqlab{A}= \mex\left(\{0, 1, 2, 3\}  \cup \{\gpg(H') \mid H_S \to H', H' \in \legalH\} \right)\\
&= \mex\left( \{0, 1, 2, 3\} \cup \{\gpg(H_{S \cup \{x\}}) \mid S \cup \{x\} \in \legalS, x \in X \setminus S\} \right)\\
&\eqlab{B}= \mex\left( \{0, 1, 2, 3\} \cup \{4 + \gat(S \cup \{x\}) \mid S \cup \{x\} \in \legalS, x \in X \setminus S\} \right)\\
&= 4 + \mex\{\gat(S \cup \{x\}) \mid S \cup \{x\} \in \legalS, x \in X \setminus S\}\\
&= 4 + \gat(S),
\end{align}
where
(A) follows from Lemma~\ref{lem:deviating}, and
(B) follows from the induction hypothesis.
\end{proof}

Let $Q = (V(Q), E(Q))$ be a hypergraph defined by $V(Q) = \{r_1, r_2, r_3\}$ and $E(Q) = \{e_1, e_2, e_3\}$ with $e_1 = \{r_1, r_2\}$, $e_2 = \{r_2, r_3\}$, and $e_3 = \{r_3\}$.
Direct evaluation confirms that $\gpg(Q) = 4$.

Consider the disjunctive sum of $H_{\emptyset}$ and $Q$;
namely, the \textsc{Two-Level Posetgame} played on $P_F = (V(P_F), E(P_F))$ defined as
$V(P_F) \coloneqq V(H_{\emptyset}) \sqcup V(Q)$ and $E(P_F) \coloneqq E(H_{\emptyset}) \sqcup E(Q)$.
Then, the outcome of \textsc{Two-Level Posetgame} on $P_F$ coincides with that of \textsc{Avoid True} on $F$ as follows.

\begin{lemma}
\label{thm:reduction_correctness}
The outcome of \textsc{Avoid True} on $F$ is equal to the outcome of \textsc{Two-Level Posetgame} on $P_F$.
\end{lemma}
\begin{proof}
We have
\begin{align}
\gpg(P_F) = 0
&\iff \gpg(H_{\emptyset}) \oplus \gpg(Q) = 0\\
&\iff \gpg(H_{\emptyset}) = 4\\
&\eqlab{A}\iff \gat(\emptyset)+4 = 4\\
&\iff \gat(\emptyset) = 0,
\end{align}
where (A) follows from Lemma \ref{lem:h_to_g}.
\end{proof}

Therefore, determining the outcome of \textsc{Avoid True} on $F$ is reduced to determining the outcome of \textsc{Two-Level Posetgame} on $P_F$.
The number of vertices and hyperedges of $P_F$ is bounded above by a polynomial in the number of variables and clauses,
so that $P_F$ can be constructed in polynomial time.
This completes the proof of Theorem \ref{thm:main}.

\subsection{Proof of Lemma~\ref{lem:deviating}}\label{sec:proof_deviating}

To establish Lemma~\ref{lem:deviating}, we analyze a slightly broader class of positions than legal positions, referred to as \emph{good} positions below.

\begin{definition}
A position $H = (V(H), E(H)) \in \subH$ is said to be \emph{good} if $e_M \in E(H)$ and $u \in V(H)$.
Let $\goodH$ denote the set of all good positions.
For $H \in \goodH$, let $k_H$ denote the number of remaining clause hyperedges of $H$,
and define $U_H \coloneqq (V(H) \cap X) \cup \{u\}$ and $p_H \coloneqq |U_H| \bmod 2$.
\end{definition}

By definition, it holds that $\legalH \subseteq \goodH$.
Note that the condition $e_M \in E(H)$ implies that all vertices in $M$ remain in $H$.

\begin{lemma}
\label{lem:four-moves}
For any $H \in \goodH$, if $k_H > 0$, then the following four types of moves are available from $H$, and the Grundy values of the resulting positions $H'$ are given in Table \ref{tab:transitions}.
\begin{itemize}
\item Removing $e_M$.
\item Removing $a \in M$ (or $b \in M$).
\item Removing some $d \in D_{i, j}$ satisfying $e_{i, j} \in E(H)$ (i.e., the clause hyperedge $e_{i, j}$ remains in $H$).
\item Removing $u$.
\end{itemize}
\end{lemma}

\begin{table}[ht]
\centering
\caption{Grundy values after the four specific types of moves, where $p \coloneqq p_H$.}
\label{tab:transitions}
\begin{tabular}{c|c|c}
\toprule
Chosen Element  & Resulting $\gpg(H')$ & Resulting $\gpg(H')$\\
& (if $k_H$ is positive and even) &  (if $k_H$ is odd) \\
\midrule
$e_M$ & $p$ & $p \oplus 3$ \\
$a$ or $b$ & $p \oplus 1$ & $p \oplus 2$ \\
$d \in D_{i, j}$ satisfying $e_{i, j} \in E(H)$ & $p \oplus 2$ & $p \oplus 1$ \\
$u$ & $p \oplus 3$ & $p \oplus 3$ \\
\bottomrule
\end{tabular}
\end{table}

\begin{proof}
For simplicity, let $k \coloneqq k_H, U \coloneqq U_H, p \coloneqq p_H$ by omitting the subscript $H$.
 These four moves are available because the assumption ensures $\{a, b, u\} \cup \left( \bigcup_{i=1}^m\bigcup_{j \in \{0, 1\}}D_{i, j} \right) \subseteq  M \cup \{u\} \subseteq V(H)$ and $e_M \in E(H)$, and the condition $k > 0$ ensures that there exists $d \in D_{i, j}$ with its clause hyperedge $e_{i, j}$ remaining in $H$.
  
We analyze the Grundy values of the positions $H'$ resulting from each of the four types of moves as follows.
\begin{itemize}
\item Removing $e_M$: The hypergraph $H'$ consists of $k$ clause hyperedges with odd cardinality and $|M| + |U|$ vertices.
By Proposition~\ref{prop:odd_formula}, we have
\begin{align}
\gpg(H') = ((|M|+|U|) \bmod 2) \oplus 3(k \bmod2) = p \oplus 3(k \bmod2).
\end{align}

\item Removing $a \in M$ (or $b \in M$):
This automatically removes $e_M$.
The hypergraph $H'$ consists of $k$ clause hyperedges with odd cardinality and $|M \setminus \{a\}| + |U|$ vertices.
By Proposition~\ref{prop:odd_formula}, we have
\begin{align}
\gpg(H') = ((|M \setminus \{a\}| + |U|) \bmod 2) \oplus 3(k \bmod 2) =  (p \oplus 1) \oplus 3(k \bmod 2).
\end{align}
   
\item Removing $d \in D_{i, j}$ satisfying $e_{i, j} \in E(H)$: 
This removes both $e_M$ and $e_{i, j}$.
The hypergraph $H'$ consists of $k - 1$ clause hyperedges with odd cardinality and $|M \setminus \{d\}| + |U|$ vertices. 
By Proposition~\ref{prop:odd_formula}, we have
\begin{align}
 \gpg(H') = ((|M \setminus \{d\}| + |U|) \bmod 2) \oplus 3((k - 1) \bmod 2) = (p \oplus 1) \oplus 3((k-1) \bmod 2).
\end{align}
   
\item Removing $u$: This removes all $k$ clause hyperedges.
The hypergraph $H'$ consists of the single hyperedge $e_M$, its vertices $M$, and $|U|-1$ isolated vertices.
It is directly seen that the single non-empty hyperedge $e_M$ together with its vertices $M$ has Grundy value $2$,
and adding $|U|-1$ isolated vertices changes the Grundy value to $2 \oplus ((|U|-1) \bmod 2) = p \oplus 3$, so that
\begin{align}
\gpg(H') = 2 \oplus (p \oplus 1) = p \oplus 3.
\end{align}
\end{itemize}

Table \ref{tab:transitions} is obtained by applying the above analysis to the case where $k$ is even and the case where $k$ is odd.
\end{proof}

The following lemma evaluates the Grundy values for the good positions.

\begin{lemma}
\label{lem:main_eval}
For any $H \in \goodH$, we have
\begin{align}
\gpg(H)
\begin{cases}
= 2 \oplus p_H &\text{if } k_H = 0,\\
\ge 4 &\text{if } k_H \text{ is even and } k_H \neq 0,\\
= p_H  &\text{if } k_H \text{ is odd.}
\end{cases}
\label{eq:dinrdj0sjsse}
\end{align}
\end{lemma}

\begin{proof}
Let $k \coloneqq k_H$, $U \coloneqq U_H$, and $p \coloneqq p_H$ for simplicity.
We proceed by induction on $|U| + k$. 

\begin{itemize}
\item The case $k = 0$:
The hypergraph $H$ consists of the single hyperedge $e_M$, its vertices $M$, and $|U|$ isolated vertices.
The single non-empty hyperedge $e_M$ together with its vertices $M$ has Grundy value $2$.
Adding $|U|$ isolated vertices changes the Grundy value to $2 \oplus (|U| \bmod 2) = 2 \oplus p$. 
Thus, $\gpg(H) = 2 \oplus p$.

\item The case where $k$ is positive and even:
  By Lemma \ref{lem:four-moves}, the four types of moves in Table~\ref{tab:transitions} are available and lead to positions with Grundy values $p, p \oplus 1, p \oplus 2, p \oplus 3$.
  Since $p \in \{0, 1\}$, we have $\{p, p \oplus 1, p \oplus 2, p \oplus 3\} = \{0, 1, 2, 3\}$, which implies $\gpg(H) \ge 4$.

\item The case where $k$ is odd:
  By Lemma \ref{lem:four-moves}, the four types of moves in Table~\ref{tab:transitions} are available and yield Grundy values $p \oplus 1, p \oplus 2, p \oplus 3$, none of which equals $p$.
  Therefore, to prove $\gpg(H) = \mex\{p \oplus 1, p \oplus 2, p \oplus 3, \ldots\} = p$, it suffices to verify that no other move leads to a position with Grundy value $p$.
  We analyze the remaining types of moves as follows, where $H'$ denotes the position resulting from the move on $H$.
  \begin{itemize}
	  \item Removing $e_{i, j}$: This leaves an even number $k - 1$ of clause hyperedges.
	  Since $H'$ remains good and $|U_{H'}| + k_{H'} = |U| + (k-1) < |U| + k$, we may apply the induction hypothesis to $H'$ to see that
	  the Grundy value $\gpg(H')$ does not equal $p$ as follows.
	  \begin{itemize}
		  \item If $k-1 > 0$, then $\gpg(H') \ge 4 > p$ by the second case of \eqref{eq:dinrdj0sjsse}.
		  \item If $k-1 = 0$, then $\gpg(H') = 2 \oplus p \neq p$ by the first case of \eqref{eq:dinrdj0sjsse}.
	  \end{itemize}
  	\item Removing $x \in U \setminus \{u\}$: This flips the parity $p$ to $p_{H'} = |U \setminus \{x\}| \bmod 2 =  p \oplus 1$.
  	 Since $H'$ remains good and $|U_{H'}| + k_{H'} = (|U|-1) + k_{H'} \leq (|U|-1) + k < |U| + k$, we may apply the induction hypothesis to $H'$ to see that the Grundy value $\gpg(H')$ is given depending on $k_{H'}$ as follows.
  	  \begin{itemize}
  	  \item If $k_{H'} = 0$, then $\gpg(H') = 2 \oplus p_{H'} = 2 \oplus (p \oplus 1) = p \oplus 3$ by the first case of \eqref{eq:dinrdj0sjsse}.
  	  \item If $k_{H'}$ is positive and even, then $\gpg(H') \ge 4$ by the second case of \eqref{eq:dinrdj0sjsse}.
  	  \item If $k_{H'}$ is odd, then $\gpg(H') = p_{H'} = p \oplus 1$ by the third case of \eqref{eq:dinrdj0sjsse}.
  	  \end{itemize}
      None of these equals $p$.
  	\item Removing $d \in D_{i, j}$ satisfying $e_{i, j} \not\in E(H)$ (i.e., the clause hyperedge $e_{i, j}$ was already removed): This behaves identically to removing $a$, yielding Grundy value $p \oplus 2 \neq p$ by Table~\ref{tab:transitions}.
  \end{itemize}
  Consequently, no move leads to a position with Grundy value $p$, while transitions to $p \oplus 1, p \oplus 2, p \oplus 3$ are always available. Thus, $\gpg(H) = \mex \{ p \oplus 1, p \oplus 2, p \oplus 3, \dots \} = p$.
\end{itemize}
\end{proof}

We are now ready to complete the proof of Lemma~\ref{lem:deviating}.

\begin{proof}[Proof of Lemma~\ref{lem:deviating}]
In a legal position $H \coloneqq H_S$, at least one clause $C_i$ is unsatisfied, and both hyperedges $e_{i,0}$ and $e_{i,1}$ associated with every unsatisfied clause remain, so that $k_H$ is positive and even.
To prove Lemma~\ref{lem:deviating}, we examine transitions from $H = H_S \in \legalH$ to illegal positions $H' \notin \legalH$ by categorizing all possible moves as follows.

\begin{itemize}
\item Removing $e_M$, $a$, $b$, $u$, or $d \in D_{i, j}$ satisfying $e_{i, j} \in E(H)$ (i.e., the four types of moves in Table~\ref{tab:transitions}):
As stated above, $k_H$ is a positive even integer. 
By Lemma \ref{lem:four-moves}, these moves transition to positions with Grundy values in $\{p_H, p_H \oplus 1, p_H \oplus 2, p_H \oplus 3\} = \{0, 1, 2, 3\}$.

\item Removing  $d \in D_{i, j}$ satisfying $e_{i, j} \not\in E(H)$:
This behaves identically to removing $a$ or $b$, yielding Grundy value $p_H \oplus 1$ by Table~\ref{tab:transitions}.

\item Removing a single clause hyperedge $e_{i, j}$: This leaves an odd number $k_H - 1$ of clause hyperedges, yielding Grundy value $p_{H'} \in \{0, 1\}$ by the third case of \eqref{eq:dinrdj0sjsse}.

\item Removing $x \in X \setminus S$ that removes all remaining clause hyperedges: This leaves $k_{H'} = 0$ clause hyperedges, yielding a Grundy value $2 \oplus p_{H'} \in \{2, 3\}$ by the first case of \eqref{eq:dinrdj0sjsse}.
\end{itemize}

Therefore, every move to an illegal position yields a Grundy value in $\{0, 1, 2, 3\}$, establishing $\{\gpg(H') \mid H \to H', H' \not\in \legalH\} \subseteq \{0, 1, 2, 3\}$.
Conversely, Lemma \ref{lem:four-moves} guarantees that by selecting $e_M$, $a$, $b$, $u$, or $d \in D_{i, j}$ with $e_{i, j} \in E(H)$,
one can transition to positions attaining every Grundy value in $\{0, 1, 2, 3\}$.
Thus, $\{\gpg(H') \mid H \to H', H' \not\in \legalH\} \supseteq \{0, 1, 2, 3\}$ also holds.
This completes the proof of Lemma~\ref{lem:deviating}.
\end{proof}

\section*{Acknowledgements}
This work was supported in part by JSPS KAKENHI Grant Number JP26K21166.
The proof of Theorem \ref{thm:main} was obtained during a session with GPT-6 Astra.

\bibliographystyle{plain}
\bibliography{draft}

@article{Sprague1935,
	author = {Sprague, Roland P.},
	title = {{\"U}ber mathematische {Kampfspiele}},
	journal = {Tohoku Mathematical Journal},
	volume = {41}, 
	pages = {438--444},
	year = {1935}
}

@article{Grundy1939,
	author = {Grundy, P. M.},
	title = {Mathematics and games},
	journal = {Eureka},
	volume = {2},
	pages = {6--8},
	year = {1939}
}

@article{Schaefer1978,
	author = {Schaefer, Thomas J.},
	title = {On the complexity of some two-person perfect-information games},
	journal = {Journal of Computer and System Sciences},
	volume = {16},
	number = {2},
	pages = {185--225},
	year = {1978}
}

@inproceedings{Grier2013,
	author = {Grier, Daniel},
	title = {Deciding the Winner of an Arbitrary Finite Poset Game Is {PSPACE}-Complete},
	booktitle = {Proceedings of the 40th International Colloquium on Automata, Languages, and Programming (ICALP 2013)},
	series = {Lecture Notes in Computer Science},
	volume = {7965},
	pages = {497--503},
	publisher = {Springer},
	year = {2013}
}

@article{FGG22,
	title={The Complexity of Poset Games},
	volume={26},
	number={1},
	journal={Journal of Graph Algorithms and Applications},
	author={Fenner, Stephen and Grier, Daniel and Gurjar, Rohit and Korwar, Arpita and Thierauf, Thomas},
	year={2022},
	month={Jan.},
	pages={1--14}
}

@article{Fenner2013,
	author = {Fenner, Stephen A. and Gurjar, Rohit and Korwar, Arpita and Thierauf, Thomas},
	title = {On two-level poset games},
	journal = {Electronic Colloquium on Computational Complexity (ECCC)},
	volume = {20},
	pages = {TR13\-019},
	year = {2013}
}

@article{OSullivan2018,
	author = {O'Sullivan, Cormac},
	title = {A vertex and edge deletion game on graphs},
	journal = {Integers},
	volume = {18},
	pages = {G03},
	year = {2018}
}

@article{Molena2022,
	author        = {Molena, T. H.},
	title         = {Take-Away Impartial Combinatorial Games on Hypergraphs and Other Related Geometric and Discrete Structures},
	year          = {2022},
	eprint        = {2203.09696},
	archiveprefix = {arXiv},
	primaryclass  = {math.CO},
}

@phdthesis{barnard2017,
	author = {Barnard, Kristen M.},
	title  = {Some Take-Away Games on Discrete Structures},
	school = {University of Kentucky},
	year   = {2017},
	type   = {{PhD} dissertation},
}

@article{Bou1901,
	ISSN = {0003486X, 19398980},
	author = {C. L. Bouton},
	journal = {Annals of Mathematics},
	number = {1/4},
	pages = {35--39},
	publisher = {[Annals of Mathematics, Trustees of Princeton University on Behalf of the Annals of Mathematics, Mathematics Department, Princeton University]},
	title = {Nim, A Game with a Complete Mathematical Theory},
	urldate = {2025-11-14},
	volume = {3},
	year = {1901}
}

@article{Gale1974,
	author = {Gale, David},
	title = {A curious Nim-type game},
	journal = {The American Mathematical Monthly},
	volume = {81},
	number = {8},
	pages = {876--879},
	year = {1974}
}

@article{Gale1982,
	author = {Gale, David and Neyman, Abraham},
	title = {Nim-type games},
	journal = {International Journal of Game Theory},
	volume = {11},
	number = {1},
	pages = {17--20},
	year = {1982}
}

@book{BCG18,
	title={Winning Ways for Your Mathematical Plays: Volume 1},
	author={Berlekamp, E. R. and Conway, J. H. and Guy, R. K.},
	isbn={9780429945595},
	year={2018},
	publisher={CRC Press}
}

\end{document}